\documentclass[10pt,twoside,onecolumn]{IEEEtran}
\usepackage {amssymb,rotating,graphicx,cite, calc, color, psfrag}
\usepackage{authblk}
 \usepackage{booktabs}
 \usepackage[table,xcdraw]{xcolor}
\usepackage[ruled,norelsize]{algorithm2e}
\usepackage{balance}

\usepackage[cmex10]{amsmath}
\usepackage{amssymb,amsthm,amsfonts,mathrsfs}
\usepackage{graphicx}
\usepackage{enumitem,stackrel}
\usepackage{color}
\usepackage{mathtools,amssymb,bm}
\usepackage{amstext}
\usepackage{array}
\usepackage{hyperref}
\theoremstyle{remark}

\ifCLASSOPTIONcaptionsoff
  \usepackage[nomarkers]{endfloat}
 \let\MYoriglatexcaption\caption
 \renewcommand{\caption}[2][\relax]{\MYoriglatexcaption[#2]{#2}}
\fi
\newcommand{\RN}[1]{%
	\textup{\uppercase\expandafter{\romannumeral#1}}%
}

\newtheoremstyle{mystyle}
  {}
  {}
  {\itshape}
  {}
  {\bfseries}
  {.}
  { }
  {}
\theoremstyle{mystyle}

\newtheorem{theorem}{Theorem}{}
\newtheorem{proposition}{Proposition}{}
{}
{}
{}
{}
\usepackage{epstopdf}

\begin{document}
%
\title{Super-resolution Method for Coherent DOA Estimation of Multiple Wideband Sources}

\author{Milad~Javadzadeh Jirhandeh, Mohammad Hossein~Kahaei \thanks{The authors are with the School of Electrical Engineering, Iran University of Science \& Technology, Tehran 16846-13114, Iran (e-mail: milad\_javadzade@elec.iust.ac.ir; kahaei@iust.ac.ir).}}%

\maketitle

\begin{abstract}
   We focus on coherent direction of arrival estimation of wideband sources based on spatial sparsity. This area of research is encountered in many applications such as passive radar, sonar, mining, and communication problems, in which an increasing attention has been devoted to improving the estimation accuracy and robustness to noise. By the development of super-resolution algorithms, narrowband direction of arrival estimation based on gridless sparse algorithms and atomic norm minimization has already been addressed. In this paper, a superresolution based method is proposed for coherent direction of arrival estimation of multiple wideband sources. We introduce an atomic norm problem by deﬁning a new set of atoms and exploiting the signal joint sparsity of different frequency subbands in a continuous spatial domain. This problem is then cast as a semideﬁnite program, which leads to implementing a new coherent direction of arrival estimation method with higher resolution and more robustness to noise. Numerical simulations show the outperformance of the proposed method compared to the conventional ones.
\end{abstract}


%
\IEEEpeerreviewmaketitle

\section{Introduction}\label{sec:Introduction}
\IEEEPARstart{W}{ideband } Direction of Arrival (DOA) estimation is popular due to its application in many fields like sonar, radar, and wireless communications. In these cases, the direction information is repeated in different frequency bands, which can be exploited to improve  the estimation precision. A conventional approach is to decompose widebands into narrowbands using specific filter banks or the Discrete Fourier Transformation (DFT) and use  the joint information of narrowband signals for DOA estimation.  These techniques can be classified into Incoherent Signal Subspace Methods (ISSM) and Coherent Signal Subspace Methods (CSSM).

In the ISSM, widebands are first divided into narrowbands and a DOA estimation method is applied to each band. The final DOA estimate is then calculated by incoherently averaging the respective results \cite{su1983signal,ferreira2012direction}.  This approach, however, suffers from a weak robustness concerning the noise and also loses its final precision when high-magnitude errors take place in narrowband DOA estimates.
In the CSSM, on the other hand, using a focusing matrix all the center frequencies of narrowband signals are mapped into a reference frequency to which a narrowband DOA estimator is applied \cite{choi2015repeated,hung1988focussing}. For instance, the Rotational Signal Subspace (RSS)\cite{hung1988focussing} may be mentioned. However, a major drawback to these methods is the need for an initial DOA estimate, which is required for designing the focusing matrix. Motivated by this deficiency, a coherent subspace method  has been addressed based on interpolation in \cite{milad}. Moreover, to  fill the gap between coherent and incoherent methods, some methods have been developed by using the signal and noise subspaces in different narrowband signals. For example, the Test of Orthogonality of Projected Subspaces (TOPS) \cite{yoon2006tops}, and Weighted Squared TOPS (WS-TOPS) \cite{hayashi2016doa} have been introduced. The main shortage of the aforementioned methods is the need for a large number of snapshots for each narrowband signal. Also, the number of DOAs should be known, a priori.

 In contrast, sparse DOA estimation methods are more practical thanks to their acceptable accuracy despite using a smaller number of snapshots. In \cite{liu2011broadband,hu2014doa}, some sparse DOA estimation methods are developed for wideband sources based on compressive sensing, in which  the DOA space is discretized by a grid for the likely values of DOAs. Then, the joint sparsity of the narrowband signals is utilized for DOA estimation by  only incorporating one snapshot of the signals. However, the main difficulty with these methods occurs for  the grid mismatch where the actual position of DOAs does not exactly lie on the grid steps. This, as a result, leads to effectively reducing the resolution of DOA estimates.

 To overcome such a difficulty in Compressed Sensing (CS), Cand$\grave{e}$s and Fernandez-Granda  extended the discrete CS to the continuous case by introducing the super-resolution concept, which directly incorporates the sparsity property in continuous domain \cite{candes2014towards}. Tang et al. \cite{tang2013compressed} used the super-resolution notion as a gridless sparse method for line spectral estimation by atomic norm minimization. Afterwards, gridless DOA estimation has been developed based on super-resolution for narrowband sources \cite{raj2018single,castanheira2019low}, which can also be applied for incoherent  DOA estimation of wideband sources.  However, such approaches exhibit inherent  disadvantages of conventional incoherent DOA estimation techniques.

In this paper, by defining a new atomic norm, we propose a coherent gridless sparse method for wideband DOA estimation based on super-resolution. The use of coherency in this method guarantees more robustness to the noise. In addition, by using only one snapshot of narrowband signals, a higher  accuracy is achieved compared to the classical methods. Also, no information  is required about the number of sources.

The paper is organized as follows. In Section \ref{sec:data model}, the signal model is defined. In Section \ref{sec:proposed}, we introduce a new atomic norm and propose the respective gridless sparse recovery problem. The performance of our method is compared to some well-known methods in Section \ref{sec:simulations} and Section \ref{sec:coclusion} concludes the results.

The following notations are respectively used in this work. Matrices and vectors are represented  by uppercase and lowercase bold letters,  ${\left(  \cdot  \right)^T}$ and ${\left(  \cdot  \right)^H}$ denote transpose and conjugate transpose operators, ${\left\| {\cdot} \right\|_F}$, ${\left\| {\cdot} \right\|_{\cal A}}$, and $\left\| {\cdot} \right\|_{\cal A}^*\;$ show the Frobenius norm, atomic norm, and dual norm of atomic norm for a matrix, $\mathop {\left\| {\cdot} \right\|}\nolimits_2$ is the $\ell _ {2} $ norm of a vector. $conv({\cdot})$  and $Tr({\cdot})$ represent the convex hull and trace operator, respectively, and $\bm{I}_m$ is an identity matrix of size $m$.
\section{Data model}\label{sec:data model}

 We consider a Uniform Linear Array (ULA) composed of $M$ omnidirectional sensors with inter-spacing $d$.  The number of sources is $K$ with the angles ${\theta _k}$, $k = 1, \ldots ,K$,  which are constant during the observation time. By applying the DFT, the array signal output is divided into $J$ narrowband signals  from ${\omega _L}$ to ${\omega _H}$. Mathematically,  we can show the array output vector in ${\omega _j}$ as
 \begin{equation}\label{formulation:1}
  {\bm{y}}\left( {{\omega _j}} \right) = {\bm{\Phi}}\left( {{\omega _j},{\boldsymbol{\theta }}} \right){\bm{s}}\left( {{\omega _j}} \right) + {\bm{n}}\left( {{\omega _j}} \right)\;\;\;\;,\;\;j = 1, \ldots ,J,
\end{equation}
where ${\bm{y}}\left( {{\omega _j}} \right) = {\left[ {{y_1}\left( {{\omega _j}} \right),\; \ldots ,{y_M}\left( {{\omega _j}} \right)} \right]^T} \in \mathbb{C}^{M\times 1}$ with ${y_m}\left( {{\omega _j}} \right)$, $m = 1, \ldots ,M$, representing the DFT of the $m$th sensor output in  ${\omega _L}\leq{\omega _j} \leq {\omega _H}$,  ${\bm{s}}\left( {{\omega _j}} \right) = {\left[ {{s_1}\left( {{\omega _j}} \right)\; \ldots ,{s_K}\left( {{\omega _j}} \right)} \right]^T} \in \mathbb{C}^{K\times 1}$ with ${s_k}\left( {{\omega _j}} \right)$ denoting the DFT of the $k$th source, ${\bm{n}}\left( {{\omega _j}} \right) \in \mathbb{C}^{M\times 1}$ is the corresponding noise, and ${\bm{\Phi}}\left( {{\omega _j},{\boldsymbol{\theta }}} \right) = {\left[ {{\bm{\phi}}\left( {{\omega _j},{\theta _1}} \right),\; \ldots ,{\bm{\phi}}\left( {{\omega _j},{\theta _K}} \right)} \right]} \in \mathbb{C}^{M\times K}$ shows the steering matrix for the DOA vector ${\boldsymbol{\theta}}  = {\left[ {{\theta _1},\; \ldots ,{\theta _K}} \right]^T}$ in ${\omega _j}$ whose columns are given by
\begin{equation}\label{formulation:2}
\begin{array}{l}
{\bm{\phi}}\left( {{\omega _j},{\theta _k}} \right) =
{\left[ {\exp \left( { - i{\omega _j}{\tau _{1,k}}} \right)\;, \ldots \;,\exp \left( { - i{\omega _j}{\tau _{M,k}}} \right)\;} \right]^T},\;
\end{array}
\end{equation}
where ${\tau _{m,k}}$ is the delay of the $k$th source with the arrival angle ${\theta _k}$ in the $m$th sensor defined as
\begin{equation}\label{formulation:3}
\begin{array}{l}
{\tau _{m,k}} = \frac{{\left( {m - 1} \right)d\sin \left( {{\theta _k}} \right)}}{c},
\end{array}
\end{equation}
with $c$ showing the wave propagation speed. We assume that ${\omega _1} > {\omega _2} >  \ldots  > {\omega _J}$ and in order to avoid ambiguity, $d$ is equal to the half of the minimum wavelength corresponding to the maximum frequency ${\omega _1}$ \cite{sellone2006robust}, that is,
\begin{equation}\label{formulation:4}
d = \frac{{{\pi c}}}{\omega _1}.
\end{equation}
Using (\ref{formulation:3}) and (\ref{formulation:4}) in (\ref{formulation:2}),  we obtain
\begin{equation}\label{formulation:5}
\begin{array}{*{20}{l}}
{{\bm{\phi}}\left( {{\omega _j},{\theta _k}} \right) = \left[ {\begin{array}{*{20}{c}}
{\begin{array}{*{20}{c}}
1\\
{\exp \left( { - i\frac{{{\omega _j}}}{{{\omega _1}}}\pi \sin \left( {{\theta _k}} \right)} \right)\;}
\end{array}}\\
{\begin{array}{*{20}{c}}
 \vdots \\
{\exp \left( { - i\left( {M - 1} \right)\frac{{{\omega _j}}}{{{\omega _1}}}\pi \sin \left( {{\theta _k}} \right)} \right)\;}
\end{array}}
\end{array}} \right].}
\end{array}
\end{equation}
Next, by defining
\begin{equation}\label{formulation:6}
\begin{array}{l}
{{\rm{\alpha }}_j} \triangleq \frac{{{\omega _j}}}{{{\omega _1}}},\\
{f_k} \triangleq \frac{1}{2}\cos \left( {{\theta _k}} \right),
\end{array}
\end{equation}
where ${{\rm{\alpha }}_j} \le 1$, ${{\rm{\alpha }}_1} = 1$ and ${\bm{f}} = {\left[ {{f_1},\; \ldots ,{f_K}} \right]^T} \in {\left[ { - \frac{1}{2},\frac{1}{2}} \right]}$ as the spatial frequencies of the DOAs, we can show
\begin{equation}
{\bm{\phi}}\left( {{\omega _j},{\theta _k}} \right) ={\bm{a}}\left( {f} \right),
\end{equation}
where
\begin{equation}\label{formulation:7}
\begin{array}{*{20}{l}}
{{\bm{a}}\left( {f} \right) = \left[
1, {\exp \left( { - i2\pi {f}} \right)\;}, \ldots, {\exp \left( { - i\left( {M - 1} \right)2\pi{f}} \right)\;}\right]^T}
\end{array}
\end{equation}
 and  $f ={\alpha _j}{f_k}$. According to (\ref{formulation:7}), we can reformulate (\ref{formulation:1}) as
\begin{equation}\label{formulation:8}
{\bm{y}}\left( {{\omega _j}} \right) = {\left[ {\bm{a}}\left( {{\alpha _j}{f_1}} \right),\; \ldots ,{\bm{a}}\left( {{\alpha _j}{f_K}} \right) \right]}
{\bm{s}}\left( {{\omega _j}} \right) + {\bm{n}}\left( {{\omega _j}} \right),
\end{equation}
from which our gridless sparse method for estimating  ${\bm{f}}$  and subsequently the DOAs of wideband sources  can be derived.

\section{Proposed gridless sparse method for wideband DOA estimation}\label{sec:proposed}
We assume that  the focusing matrices ${{\bm{T}}_j} \in \mathbb{C}^{M\times M}$, $j = 1, \ldots ,J$, satisfy the following properties,
\begin{equation}\label{formulation:9}
\begin{array}{l}
{\bm{a}}\left( {{{\rm{\alpha }}_j}f} \right) = {{\bm{T}}_j}{\bm{a}}\left( {{{\rm{\alpha }}_1}f} \right) + {{\bm{e}}_j}\left( f \right)\\
\;\;\;\;\;\;\;\;\;\;\;\;\;\;\; = {{\bm{T}}_j}{\bm{a}}\left( f \right) + {{\bm{e}}_j}\left( f \right),
\end{array}
\end{equation}
in which ${{\bm{e}}_j}\left( f \right)$ is the focusing error vector in ${\omega _j}$ for the DOA spatial frequency $f\in \left[ { - \frac{1}{2},\frac{1}{2}} \right]$.  This model shows that ${\bm{a}}\left( {{{\rm{\alpha }}_j}f} \right)$ can linearly be approximated by ${\bm{a}}\left( f \right)$ and ${{\bm{T}}_j}$, in the sense that ${\bm{a}}\left( f \right)$ with the spatial frequency $f$  can be focused on the lower frequencies, ${{\alpha }_j}f$. This lets  us to  design a more accurate and efficient focusing matrix as follows.\\
To generate the focusing matrices ${{\bm{T}}_j}$, we use the method developed in \cite{milad} to get
\begin{equation}\label{formulation:10}
\begin{array}{l}
{{\bm{T}}_j} = \left[ {\begin{array}{*{20}{c}}
{{{\bm{T}}_j}\left( {1,1} \right)}& \ldots &{{{\bm{T}}_j}\left( {1,M} \right)}\\
 \vdots & \ddots & \vdots \\
{{{\bm{T}}_j}\left( {M,1} \right)}& \ldots &{{{\bm{T}}_j}\left( {M,M} \right)}
\end{array}} \right],\\
\end{array}
\end{equation}
where  ${{\bm{T}}_j}\left( {m,m'} \right) = sinc\left( {{{\rm{\alpha }}_j}\left( {m - 1} \right) - \left( {m' - 1} \right)} \right)$, and
$m,m' = 1, \ldots ,M.$ Note that for ${\alpha _1} = 1$, we get ${{\bm{T}}_1} = {{\bm{I}}_M}$.

Using the output vectors of adjacent sensors  given by (\ref{formulation:8}) and using  (\ref{formulation:9}), the array output matrix $
{\bm{Y}} = \left[ {{\bm{y}}\left( {{\omega _1}} \right), \ldots ,{\bm{y}}\left( {{\omega _J}} \right)} \right] \in\mathbb{C}^{M\times J}$ is
obtained as
\begin{equation}\label{formulation:13}
{\bm{Y}} = \mathop \sum \limits_{k = 1}^K {\beta _k}{\bm{A}}\left( {{f_k},{{\bm{c}}_k}} \right) + {\bm{N}} + {\bm{E}},
\end{equation}
where  $ {\bm{A}}\left( {f_k,{\bm{c}_k}} \right) = \left[ {{{\bm{T}}_1}{\bm{a}}\left( f_k \right), \ldots ,{{\bm{T}}_J}{\bm{a}}\left( f_k \right)} \right]\times diag\left( {\bm{c}_k} \right)$, ${\beta _k}$ and ${{\bf{c}}_k}$ are defined according to (\ref{formulation:8}) as
$$\begin{array}{l}
{\beta _k} = \;\mathop {\left\| {\left[ {\begin{array}{*{20}{c}}
{{s_k}\left( {{\omega _1}} \right)}\\
 \vdots \\
{{s_k}\left( {{\omega _J}} \right)}
\end{array}} \right]} \right\|}\nolimits_2 ,\;\;{{\bm{c}}_k} =  {\left[ {\begin{array}{*{20}{c}}
{{s_k}\left( {{\omega _1}} \right)}\\
 \vdots \\
{{s_k}\left( {{\omega _J}} \right)}
\end{array}} \right]} /{\beta _k},\\
\end{array} $$
and the noise  and focusing error matrices  ${\bm{N}}$ and ${\bm{E}}$ are respectively given by
$$\begin{array}{l}
{\bm{N}} = \left[ {{\bm{n}}\left( {{\omega _1}} \right), \ldots ,{\bm{n}}\left( {{\omega _J}} \right)} \right]\in\mathbb{C}^{M\times J},\\
\;{\bm{E}} = \left[ {{{\bm{\acute{e}}}_1}, \ldots ,{{\bm{\acute{e}}}_J}} \right]\in\mathbb{C}^{M\times J},\\
\;{{\bm{\acute{e}}}_j} = \left[ {{{\bm{e}}_j}\left( {{f_1}} \right),\; \ldots ,{{\bm{e}}_j}\left( {{f_K}} \right)} \right] \times {\bm{s}}\left( {{\omega _j}} \right).
\end{array}$$

In this way, the noiseless array output matrix is
${\bm{X^\star}} = \mathop \sum\nolimits_{k=1}^{K} {\beta _k}{\bm{A}}\left( {{f_k},{{\bm{c}}_k}} \right),$
which can be recovered from the array output matrix in order to estimate the DOA spatial frequencies vector $\bm f$. For this purpose, we define the set of atoms as
$$ \begin{array}{l}
{\cal A} = \{{\bm{A}}\left( {f,{\bm{c}}} \right) = \left[ {{{\bm{T}}_1}{\bm{a}}\left( f \right), \ldots ,{{\bm{T}}_J}{\bm{a}}\left( f \right)} \right]\times diag\left( {\bm{c}} \right){\rm{|}}\;\\
\;\;\;\;\;\;\;\;\;\;f \in \left[ { - \frac{1}{2},\frac{1}{2}} \right],{\bm{c}} \in\mathbb{C}^{J\times 1} ,\;{\left\| {\bm{c}} \right\|_2} = 1\},
\end{array}$$
by which the atomic norm is defined for ${\bm{X}}$ as
\begin{equation}\label{formulation:12}
\begin{array}{l}
{{\left\| {\bm{X}} \right\|}_{\cal A}} = {\rm{inf}}\{ t > 0{\rm{\;}}:{\bm{X}} \in t\;conv\left( {\cal A} \right)\} \\
  \begin{array}{l}
\;\;\;\;\;\;\;\;\;= \mathop {\inf }\limits_{{f_k},{\beta _k},{{\bm{c}}_k}} \{ \mathop \sum \nolimits_k {\beta _k}\;:{\bm{X}} = \mathop \sum \nolimits_k {\beta _k}{\bm{A}}\left( {{f_k},{{\bm{c}}_k}} \right),\;\\
\;\;\;\;\;\;\;\;\;\;\;\;\;\;\;\;\;\;\;\;\;\;\;\;{f_k} \in \left[ { - \frac{1}{2},\frac{1}{2}} \right],\;\left\|{{\bm{c}}_k}\right\|_2 = 1, {\beta _k} > 0\}.
\end{array}
\end{array}
\end{equation}
Assuming that the sum of the noise power and focusing error power is equal to
$\left\| {\bm{N}} \right\|_F^2 + \left\| {\bm{E}} \right\|_F^2 = \gamma $ and that the number of sources is small, we propose the following sparse problem for coherent estimation of wideband DOAs,
\begin{equation}\label{formulation:15}
\begin{array}{l}
\mathop {\min }\limits_{ {\bm{X}}} \;{\left\| { {\bm{X}}} \right\|_{\cal A}}\\
subject\;to\;\;{\left\| { {\bm{X}} - {\bm{Y}}} \right\|_F} \le \;\sqrt \gamma.
\end{array}
\end{equation}
The optimum matrix ${{\bm{X}}_{opt}}$ resulted from (\ref{formulation:15}); which is an estimate of the noiseless data matrix ${\bm{X^\star}}$, can be described by its atoms as
\begin{equation}\label{formulation:16}
{ {\bm{X}}_{opt}} = \mathop \sum \limits_{k = 1}^{\widehat K} {\widehat \beta _k}{\bm{A}}\left( {{{\widehat f}_k},{{\widehat {\bm{c}}}_k}} \right),
\end{equation}
where ${\widehat \beta _k}$ ,${{\widehat f}_k}$, and ${{\widehat {\bm{c}}}_k}$ are the estimates of true ${\beta _k}$,${{f}_k}$, and ${{\bm{c}}_k}$, respectively, and ${\widehat K}$ is an estimate of real source numbers ${K}$.

To solve the primal problem in (\ref{formulation:15}) and estimate the DOAs from the atoms of $ {\bm{X}}_{opt}$, we present its Lagrangian dual problem as
\begin{equation}\label{formulation:17}
\begin{array}{l}
\mathop {{\rm{m}}ax}\limits_{\bm{H}} Re\left\{ {Tr\left( {{{\bm{Y}}^H}{\bm{H}}} \right)} \right\} - \sqrt \gamma  {\left\| {\bm{H}} \right\|_F}\\
subject\;to\; \;\left\| {\bm{H}} \right\|_{\cal A}^* \le 1\;
\end{array},
\end{equation}
where $\left\| \cdot  \right\|_{\cal A}^*\;$ shows the dual atomic norm. Since $ {\bm{X}} = {\bm{Y}}$ is a feasible solution for  (\ref{formulation:15}), strong duality holds according to Slater's condition \cite{boyd2004convex}.

With the assumption of  ${\bm{H}} = \left[ {{{\bm{h}}_1}, \ldots ,{{\bm{h}}_J}} \right] \in \mathbb{C}^{M\times J}$, the dual atomic norm $\left\| {\bm{H}} \right\|_{\cal A}^*$ in (\ref{formulation:17}) is defined as
\begin{equation}\label{formulation:18}
\begin{array}{l}
\left\| {\bm{H}} \right\|_{\cal A}^* = \mathop {\sup }\limits_{\;{{\left\| {\bm{X}} \right\|}_{\cal A}} \le 1} Re\left\{ {Tr\left( {{{\bm{H}}^H}{\bm{X}}} \right)} \right\}\\
\;\;\;\;\;\;\;\;\;\;\; = \mathop {\sup }\limits_{f \in \left[ { - \frac{1}{2},\frac{1}{2}} \right]\;,\;\;\;\left\|{{\bm{c}}_k}\right\|_2 = 1} Re\left\{ {Tr\left( {{{\bm{H}}^H}{\bm{A}}\left( {f,{\bm{c}}} \right)} \right)} \right\}\\
\;\;\;\;\;\;\;\;\;\;\; = \mathop {\sup }\limits_{f \in \left[ { - \frac{1}{2},\frac{1}{2}} \right]\;,\;\;\;\left\|{{\bm{c}}_k}\right\|_2 = 1} Re\left\{ {{{\bm{c}}^T}\left[ {\begin{array}{*{20}{c}}
{{{({\bm{T}}_1^H{{\bm{h}}_1})}^H}{\bm{a}}\left( f \right)}\\
 \vdots \\
{{{({\bm{T}}_J^H{{\bm{h}}_J})}^H}{\bm{a}}\left( f \right)}
\end{array}} \right]} \right\}\\
\;\;\;\;\;\;\;\;\;\;\; = \mathop {\max }\limits_{f \in \left[ { - \frac{1}{2},\frac{1}{2}} \right]\;} {\left\| {\left[ {\begin{array}{*{20}{c}}
{{{({\bm{T}}_1^H{{\bm{h}}_1})}^H}{\bm{a}}\left( f \right)}\\
 \vdots \\
{{{({\bm{T}}_J^H{{\bm{h}}_J})}^H}{\bm{a}}\left( f \right)}
\end{array}} \right]} \right\|_2},
\end{array}
\end{equation}
where  for the last expression, we have used the Cauchy-Schwarz inequality $\begin{array}{l}
 Re\left\{ {{{\bm{p}}^H} {{\bm{q}}}} \right\}\leq {\left\| \bm{p} \right\|_2}{\left\| \bm{q} \right\|_2},\\
\end{array}$
which is held for any vectors $\bm{p}$ and $\bm{q}$ of the same size.
Next, we describe the primal in  (\ref{formulation:17}) in the form of a Semidefinite Programming (SDP). For this purpose, the condition $\left\| {\bm{H}} \right\|_{\cal A}^* \le 1$ should be expressed in the form of an SDP condition, for which we present the following proposition.
\begin{proposition}\label{proposition}
For matrices ${\bm{H}} = \left[ {{{\bm{h}}_1}, \ldots ,{{\bm{h}}_J}} \right] \in \mathbb{C}^{M\times J}$, $\overline {\bm{H}}  = \left[ {{\bm{T}}_1^H{{\bm{h}}_1}, \ldots ,{\bm{T}}_J^H{{\bm{h}}_J}} \right]\in \mathbb{C}^{M\times J}$, and ${{\bm{T}}_j} \in \mathbb{C}^{M\times M}$, $j = 1, \ldots ,J$, if there exists a Hermitian matrix ${\bm{Q}}\in \mathbb{C}^{M\times M}$ with the condition
\begin{equation}\label{formulation:33}
\begin{array}{l}
\left[ {\begin{array}{*{20}{c}}
{\bm{Q}}&{\overline {\bm{H}} }\\
{{{\overline {\bm{H}} }^H}}&{{{\bm{I}}_J }}
\end{array}} \right]\succeq0\\\\
\end{array}
\end{equation}
and
\begin{equation}\label{formulation:34}
\sum\limits_{n = 1}^{M - m} {\mathop Q\nolimits_{n,n + m} }  = \begin{cases}
1, & m = 0\\
0, & m = 1,...,M - 1,
\end{cases}
\end{equation}
the inequality
\begin{equation}\label{formulation:19}
\left\| {\bf{H}} \right\|_{\cal A}^* = \mathop {\max }\limits_{f \in \left[ { - \frac{1}{2},\frac{1}{2}} \right]\;} {\left\| {{{{\overline {\bf{H}} }^H}{\bm{a}}(f)}} \right\|_2} \le 1
\end{equation}
is satisfied. Reciprocally, the relationship (\ref{formulation:19}) implies (\ref{formulation:33}) and (\ref{formulation:34}).
\end{proposition}
\begin{proof}
With the Schur complement, (\ref{formulation:33}) holds, if and only if,
$${\bm{Q}}\succeq0$$ and $${\bm{Q}} - {\overline {\bm{H}} }\;{{\overline {\bm{H}} }^H}  \succeq0.$$
Thus, for any ${\bm{a}}(f)$ and $f \in {\left[ { - \frac{1}{2},\frac{1}{2}} \right]}$, we get
 $$\mathop {\left\| {{{\overline {\bm{H}} }^H}{\bm{a}}(f)} \right\|}\nolimits_2^2 \leq {{\bm{a}}(f)^H} {\bm{Q}}{\bm{a}}(f).$$
But from (\ref{formulation:34}), we have ${{\bm{a}}(f)^H} {\bm{Q}}{\bm{a}}(f) = 1$ and subsequently (\ref{formulation:19}) is  proved.
The reciprocal proof is derived by backward reasoning. \\
\end{proof}
 Now, using \textit{Proposition \ref{proposition}}, the optimisation problem in (\ref{formulation:17}) can be represented by the following SDP problem,\\
\begin{equation}\label{formulation:21}
\begin{array}{l}
\mathop {\max }\limits_{\bm{H, Q}} Re\left\{ {Tr\left( {{{\bm{Y}}^H}{\bm{H}}} \right)} \right\} - \sqrt \gamma  {\left\| {\bm{H}} \right\|_F}\\\\
subject \; to \;\;\;\left[ {\begin{array}{*{20}{c}}
{\bm{Q}}&{\overline {\bm{H}} }\\\\
{{{\overline {\bm{H}} }^H}}&{{{\bm{I}}_J}}
\end{array}} \right]\succeq0,\\\\
\;\;\;\;\;\;\;\;\;\;\;\;\;\;\;\;\;\;\;\;\; \overline {\bm{H}}  = \left[ {{\bm{T}}_1^H{{\bm{h}}_1}, \ldots ,{\bm{T}}_J^H{{\bm{h}}_J}} \right],\\\\
\;\;\;\;\;\;\;\;\;\;\;\;\;\;\;\;\;\;\;\;\; \sum\limits_{n = 1}^{M - m} {\mathop Q\nolimits_{n,n + m} }  = \begin{cases}
1, & m = 0,\\
0, & m = 1,...,M - 1,
\end{cases}\\\\
\;\;\;\;\;\;\;\;\;\;\;\;\;\;\;\;\;\;\;\;\; {\bm{Q}}\; is\; Hermitian.
\end{array}
\end{equation}
From (\ref{formulation:21}), the optimum solutions
${{\bm{H}}_{opt}} = \left[ {{{\bm{h}}_{{\rm{ }}{1_{opt}}}}, \ldots ,{{\bm{h}}_{{\rm{ }}{J_{opt}}}}} \right]$ and ${\overline {\bm{H}} _{opt}} = \left[ {{\bm{T}}_1^H\mathop {\bm{h}}\nolimits_{\mathop 1\nolimits_{opt} } , \ldots ,{\bm{T}}_J^H\mathop {\bm{h}}\nolimits_{\mathop J\nolimits_{opt} } } \right]$ are obtained and DOAs are subsequently estimated based on \textit{Theorem \ref{theorem}}.
\begin{theorem}\label{theorem}
If ${ {\bm{X}}_{opt}}$ and ${\overline {\bm{H}} _{opt}}$ are the solutions of the primal  and dual problems in (\ref{formulation:15}) and  (\ref{formulation:21}), respectively,  then the estimates of DOA frequencies in (\ref{formulation:16}), ${\widehat f_k}$, and $k = 1, \ldots ,\widehat K$ will satisfy,
\begin{equation}\label{formulation:300}
 {\left\| {\left[
{\mathop {\mathop {\overline {\bm{h}} }\nolimits_1^H }\nolimits_{opt} {\bm{a}}\left( {{{\widehat f}_k}} \right)},
 \ldots,
{\mathop {\mathop {\overline {\bm{h}} }\nolimits_J^H }\nolimits_{opt} {\bm{a}}\left( {{{\widehat f}_k}} \right)}
 \right]^T} \right\|_2} = 1,
  \end{equation}
  and
  \begin{equation}\label{formulation:301}
  {\widehat {\bm{c}}_k} = {\left[ {\overline {\bm{h}} {{_1^H}_{opt}}{\bm{a}}\left( {{{\hat f}_k}} \right), \ldots ,\overline {\bm{h}} {{_J^H}_{opt}}{\bm{a}}\left( {{{\hat f}_k}} \right)} \right]^H}.
  \end{equation}
\end{theorem}
\noindent The proof of \textit{Theorem \ref{theorem}} is given in Appendix A.

In this way, by estimating ${ {\bm{X}}_{opt}}$, its atoms are used to compute ${{{\widehat f}_k}}$, as the estimates of $f_k$, and afterwards the DOAs of wideband sources are found.
\section{Numerical simulations}\label{sec:simulations}
 We consider an underwater scenario with a ULA composed of $M=16$ hydrophones and $c=1500 m/s$. Source signals are random waves with 512 samples, whose bandwidths in discrete frequency domain lie in $[\pi/3,2\pi/3]$.
Also, a 60-point DFT is applied to the received signals, where the number of the selected
frequency bins is $J = 10$, and the measuring noise in (\ref{formulation:6}) is zero mean white Gaussian with variance ${{\sigma _{noise}^2}}$.

The performance of the WGS algorithm is compared to that of the RSS and WS-TOPS methods.

The initial values for the
RSS are the true DOAs added up with some errors within $\pm \mathop 2\nolimits^ \circ$ randomly chosen from a uniform distribution.
 Moreover,  as a required information, we provide the true number of sources for both RSS and WS-TOPS methods.  Simulation results are presented by averaging $100$ independent trials of each experiment.

In the first experiment,  we consider three sources located at ${\boldsymbol{\theta}}=[{-5}^{\circ},15^{\circ},40^{\circ}]^T$.
The results are compared in Fig. \ref{figure1} in RootMean-Squared Error (RMSE) sense at different SNRs. As seen, the RSS generates the largest RMSE, mainly  due to the impact of error on the initial values which could dominate the noise effect. Furthermore, WS-TOPS achieves a better performance and WGS offers the least RMSE.
\begin{figure}[!t]
\centering
\includegraphics[width=85mm]{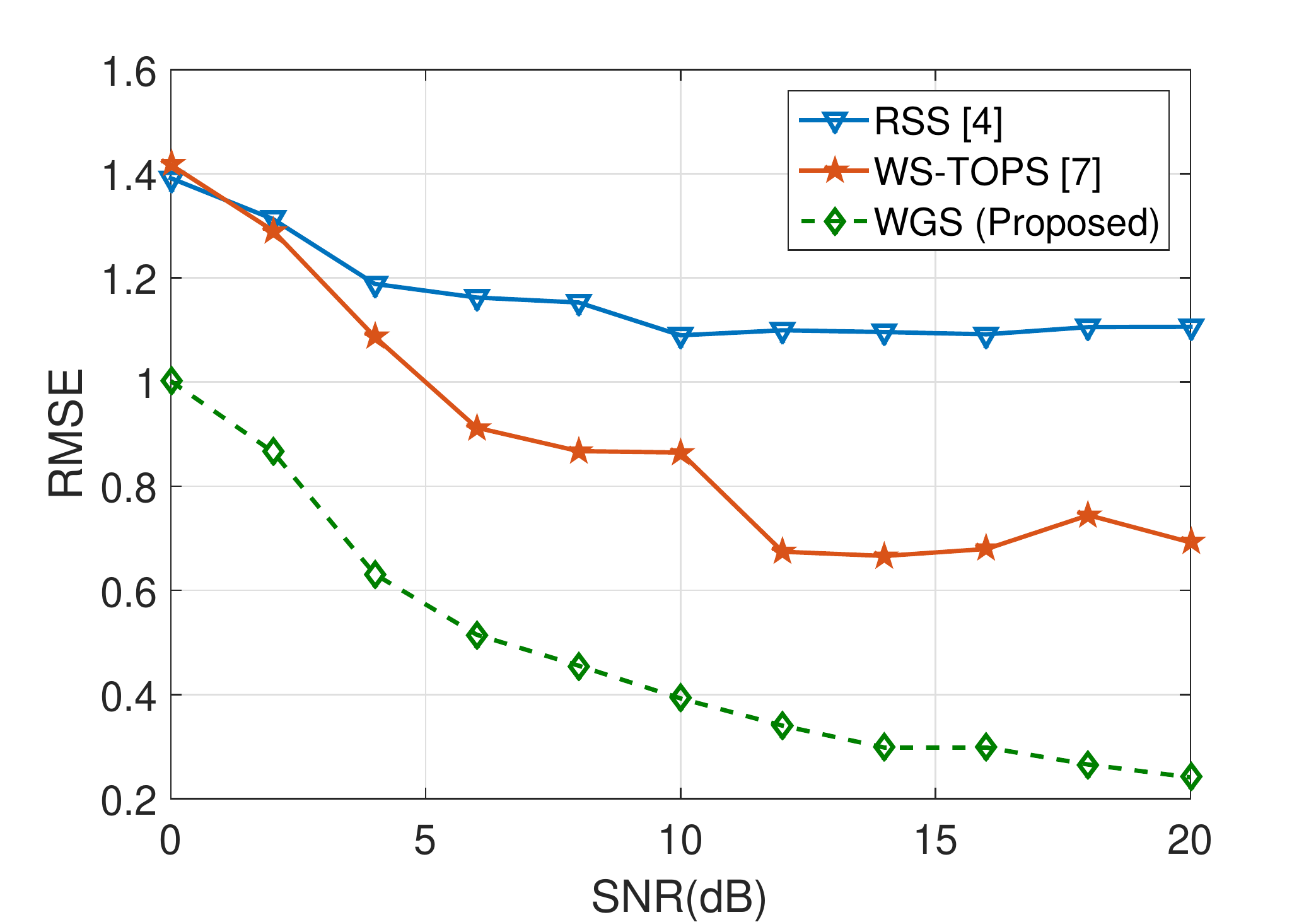}
\caption{RMSEs for WGS, RSS, and WS-TOPS methods for three sources at ${\boldsymbol{\theta}}=[{-5}^{\circ},15^{\circ},40^{\circ}]^T$ and different SNRs.}\label{figure1}
\end{figure}
  In the next experiment, we investigate the resolution of the estimators for different DOA angles at 10 dB SNR. The first DOA is fixed at ${\theta _1} = \mathop {40}\nolimits^ \circ$ and the second DOA varies between ${\theta _2} = \left[ {\mathop {28}\nolimits^ \circ  ,\mathop {37}\nolimits^ \circ  } \right]$.  The RMSEs are shown for $\Delta \theta  = {\theta _1} - {\theta _2}$ in Table \ref{tabel:2}. One can see that the WS-TOPS is unable to
estimate the DOAs for $\Delta \theta < \mathop {11}\nolimits^ \circ$, which is due to  needing more snapshots for a better performance. These values  are less than $\mathop {6}\nolimits^ \circ$ and $\mathop {3}\nolimits^ \circ$ for the RSS and WGS, respectively.

  \begin{table}[!t]
\caption{RMSEs for WGS, RSS, and WS-TOPS methods for different resolutions.}
\begin{center}
\begin{tabular}{|c||c|c|c|c|c|}\hline
$\Delta \theta$ (Degree) & 12  & 11  & 10  & 9  & 8  \\ \hline
WGS (Proposed) &\bf 0.5662  &\bf 0.5460 &\bf 0.6209 &\bf 0.7315  &\bf 0.7490  \\ \hline
WS-TOPS \cite{hayashi2016doa}  & 0.6768  & 0.6256 & Failed & Failed       & Failed   \\ \hline
RSS \cite{hung1988focussing}   & 1.0631  & 1.0838 & 1.1266 & 1.1666 & 1.1066  \\ \hline \hline
$\Delta \theta$ (Degree) & 7  & 6  & 5  & 4  & 3  \\ \hline
WGS (Proposed) &\bf 0.7510  &\bf 0.8984 &\bf 1.1119 &\bf 1.1117  &\bf 1.2242  \\ \hline
WS-TOPS \cite{hayashi2016doa}  & Failed  & Failed & Failed & Failed       & Failed   \\ \hline
RSS \cite{hung1988focussing}   & 1.0526  & 1.0728 & Failed & Failed & Failed  \\ \hline
\end{tabular}\label{tabel:2}
\end{center}
\end{table}

\section{Conclusion}\label{sec:coclusion}
We proposed a coherent gridless sparse method for wideband DOA estimation by defining a new atomic norm and solving the corresponding  SDP. Simulations results demonstrate that this method is more robustness to noise with a better resolution compared to the RSS and WS-TOPS methods. Moreover, this super-resolution based method needs no knowledge about the number of sources.

%
%


\appendix
\subsection{Proof of \textit{Theorem \ref{theorem}}}
From the strong duality theorem for the optimal solutions of primal and dual problems, ${\bm{X} _{opt}}$ and ${\bm{H} _{opt}}$ in (\ref{formulation:15}) and (\ref{formulation:21}), respectively,  we get
\renewcommand{\theequation}{\thesubsection.\arabic{equation}}
\setcounter{equation}{0}
\begin{equation}\label{formulation:22}
\begin{array}{l}
\mathop {\left\| {\mathop { {\bm{X}}}\nolimits_{opt} } \right\|}\nolimits_{\cal A} \; = Re\left\{ {Tr\left( {{{\bm{Y}}^H}{{\bm{H}}_{opt}}} \right)} \right\} - \sqrt \gamma  {\left\| {\bm{H}_{opt}} \right\|_F}\\
 = Re\left\{ {Tr\left( {\mathop { {\bm{X}}}\nolimits_{opt}^H {{\bm{H}}_{opt}}} \right)} \right\}\\ + Re\left\{ {Tr\left( {{{\left( {{\bm{Y}} - \mathop { {\bm{X}}}\nolimits_{opt} } \right)}^H}{{\bm{H}}_{opt}}} \right)} \right\} - \sqrt \gamma  {\left\| {\bm{H}_{opt}} \right\|_F}\\
 \le Re\left\{ {Tr\left( {\mathop { {\bm{X}}}\nolimits_{opt}^H {{\bm{H}}_{opt}}} \right)} \right\}.
\end{array}
\end{equation}
The last inequality has been written using the constraint in (\ref{formulation:15}), where we have ${\left\| {{\left( {{\bm{Y}} - \mathop { {\bm{X}}}\nolimits_{opt} } \right)}} \right\|_F} \leq \sqrt \gamma$, and incorporating
the following Cauchy-–Schwarz inequality,
$$ Re\left\{ {Tr\left( {{{\left( {{\bm{Y}} - \mathop { {\bm{X}}}\nolimits_{opt} } \right)}^H}{{\bm{H}}_{opt}}} \right)} \right\}\leq {\left\| {{\left( {{\bm{Y}} - \mathop { {\bm{X}}}\nolimits_{opt} } \right)}} \right\|_F}{\left\| {\bm{H}_{opt}} \right\|_F}.$$\\
 On the other hand, from the dual norm definition, we can write,
\begin{equation}\label{formulation:23}
\begin{array}{l}
Re\left\{ {Tr\left( {\mathop { {\bm{X}}}\nolimits_{opt}^H {{\bm{H}}_{opt}}} \right)} \right\} \le \;\mathop {\left\| {\mathop { {\bm{X}}}\nolimits_{opt} } \right\|}\nolimits_{\cal A} \mathop {\left\| {{{\bm{H}}_{opt}}} \right\|}\nolimits_{\;{\cal A}}^* \\
\;\;\;\;\;\;\;\;\;\;\;\;\;\;\;\;\;\;\;\;\;\;\;\;\;\;\;\;\;\;\;\;\;\;\; \le \;\mathop {\left\| {\mathop { {\bm{X}}}\nolimits_{opt} } \right\|}\nolimits_{\cal A},
\end{array}
\end{equation}
and thus from (\ref{formulation:22}) and (\ref{formulation:23}), the relationship
$$\mathop {\left\| {\mathop { {\bm{X}}}\nolimits_{opt} } \right\|}\nolimits_{\cal A}  = Re\left\{ {Tr\left( {\mathop { {\bm{X}}}\nolimits_{opt}^H {{\bm{H}}_{opt}}} \right)} \right\},$$ holds. Using (\ref{formulation:16})  and the Cauchy–-Schwarz inequality, the latter expression  leads to \begin{equation}\label{formulation:24}
\begin{array}{l}
{\left\| {{{ {\bm{X}}}_{opt}}} \right\|_A} \\
 =Re\left\{ Tr\left( \mathop {\bm{H}}\nolimits_{opt}^H \sum\limits_{k = 1}^{\hat K} {{{\hat \beta }_k}} [{\hat  c_1}\mathop T\nolimits_1 {\bm{a}}({\hat f _k}),...,{\hat c_J}\mathop T\nolimits_J {\bm{a}}({\hat  f_k})]\right) \right\}  \\
 =\mathop \sum \limits_{k = 1}^{\widehat K} {\widehat \beta _k}Re\left\{ {\widehat {\bm{c}}_k^T\left[ {\begin{array}{*{20}{c}}
{{{({\bm{T}}_1^H\mathop {\bm{h}}\nolimits_{\mathop 1\nolimits_{opt} } )}^H}{\bm{a}}\left( {{{\widehat f}_k}} \right)}\\
 \vdots \\
{{{({\bm{T}}_J^H\mathop {\bm{h}}\nolimits_{\mathop J\nolimits_{opt} } )}^H}{\bm{a}}\left( {{{\widehat f}_k}} \right)}
\end{array}} \right]} \right\}
 \le \mathop \sum \limits_{k = 1}^{\widehat K} {\widehat \beta _k},
\end{array}
\end{equation}
in which we incorporated (\ref{formulation:17} and (\ref{formulation:18}) to show that,
 $$\left\|{\left[
{{{({\bm{T}}_1^H\mathop {\bm{h}}\nolimits_{\mathop 1\nolimits_{opt} } )}^H}{\bm{a}}\left( {{{\widehat f}_k}} \right)}, \ldots, {{{({\bm{T}}_J^H\mathop {\bm{h}}\nolimits_{\mathop J\nolimits_{opt} } )}^H}{\bm{a}}\left( {{{\widehat f}_k}} \right)} \right]^T}\right\|_2 \leq 1.$$
Since from (\ref{formulation:16}), we obtain ${\left\| {{{ {\bm{X}}}_{opt}}} \right\|_A} = \mathop \sum \limits_{k = 1}^{\widehat K} {\widehat \beta _k}$,  (\ref{formulation:24}) is only satisfied for the equalities (\ref{formulation:300}, and (\ref{formulation:301}, which complete the proof.

\bibliographystyle{ieeetr}
\bibliography{Milad_Main_References}
\balance

\end{document}